\documentclass[aps,pra,twocolumn,nofootinbib,longbibliography]{revtex4-2}

\usepackage[T1]{fontenc}
\usepackage{lmodern}
\usepackage{amsmath,amssymb,braket,bm}
\usepackage{graphicx}
\usepackage{hyperref}
\usepackage{quantikz}          
\hypersetup{colorlinks=true,allcolors=blue}
\usepackage{comment}
\usepackage{amsthm}
\newtheorem{proposition}{Proposition}
\newtheorem{theorem}{Theorem}

\usepackage{etoolbox}
\makeatletter
\AtBeginEnvironment{thebibliography}{%
  \renewcommand{\href}[2]{#2}
  \renewcommand{\url}[1]{#1}
}
\makeatother

\tikzset{
  slashbox/.style={
    draw,
    append after command={
      (\tikzlastnode.south west) -- (\tikzlastnode.north east)
    }
  }
}

\begin{document}

\title{Quantum observers can communicate across multiverse branches}
\author{Maria Violaris}
\email{maria@violaris.com}
\affiliation{Department of Physics, University of Oxford, Oxford, United Kingdom}
\date{September 2026}

\begin{abstract}
It is commonly thought that observers in distinct branches of an Everettian multiverse cannot communicate without violating the linearity of quantum theory. Here we show a counterexample, demonstrating that inter-branch communication is in fact possible, entirely within standard quantum theory. We do this by considering a Wigner's-friend scenario, where an observer (Wigner) can have quantum control over another observer (the friend). We present a thought experiment where the friend in superposition can receive a message written by a distinct copy of themselves in the multiverse, with the aid of Wigner. The protocol involves writing the message within a branch, then Wigner running an uncomputation, branch-swap operation, and recomputation of the friend. Our thought experiment challenges conventional wisdom regarding the ultimate limits of what is possible in an Everettian multiverse. It has a surprising potential application which involves using knowledge-creation paradoxes for testing Everettian quantum theory against single-world theories. 
\end{abstract}

\maketitle

\section{Introduction}\label{sec:intro}

When observers are treated as quantum systems, many apparent paradoxes and ambiguities are resolved, leading to a counterintuitive but self-consistent Everettian multiverse account of reality \cite{Everett1957,Deutsch2002,Wallace2012}. In this account, while individual observers experience an approximately classical macroscopic reality locally, on a global scale they exist in superpositions of multiple such experiences occurring on different quantum ``branches'' \cite{Zurek2003}. A natural question to consider in such a setting is whether observers in a quantum multiverse can communicate across their branches, for instance to exchange a message with a copy of themselves in a different branch of the multiverse where they are having a different local experience. It is commonly thought that communication across branches of the multiverse is universally forbidden as it would violate the linearity of quantum theory \cite{Deutsch1997,Polchinski1991}. Here we challenge that notion with an explicit protocol that enables exchange of messages between observers across branches, within fully unitary, linear quantum theory.

The impossibility of inter-branch communication appears to follow intuitively from decoherence theory \cite{Zurek2003,Schlosshauer2005}. Once an observer measures a quantum system in superposition, and then interacts with a complex environment, then the global superposition of the system, observer and environment branches into two approximately independently evolving worlds \cite{Zurek2003}. The emergence of these approximately independent branches is now well-established, though there is continuing research on various aspects of the emergence of apparent classicality (e.g. \cite{Riedel2017}). To achieve inter-branch communication within standard linear, unitary quantum theory, we drop the assumption that the decohered observer branches cannot be coherently controlled. This assumption does not hold in so-called ``Wigner's friend'' scenarios, a family of thought experiments which were originally used to investigate the measurement problem in quantum mechanics \cite{Wigner1961}. These have since become a popular tool in the quantum foundations community for formulating and testing core assumptions of quantum theory and potential successor theories (e.g. \cite{Frauchiger2018,Brukner2018,Bong2020}). We note that previous studies have considered protocols by which Wigner and his friend may exchange information \cite{baumann2020wigner, del2025wigner}, whereas here we study information exchange between the friends' branches, with Wigner merely acting as the agent who enables it.

Our thought experiment to demonstrate the possibility of inter-branch communication works as follows. The protocol involves an observer inside an isolated laboratory subject to global control, and transports classical $n$-bit messages between two versions of the friend in superposition. It works by (i)~creating a superposition of the observer, (ii)~letting the observers in each branch compose a message and write it down, (iii)~uncomputing the observers' message-writing dynamics, so that the messages survive only on paper, (iv)~applying a global unitary that exchanges the observer branches, and (v)~recomputing the message-writing dynamics. Then each observer receives the message originally written by their counterpart in the other branch. We present a unitary implementation of the protocol using a simple quantum circuit toy-model, and discuss how issues of knowledge-creation mean that this protocol provides a novel test of many-worlds against single-world accounts of quantum theory. 

\section{Protocol for inter-branch communication}\label{sec:memory-protocol}

Here we present a protocol in which friends in superposition compose messages, and each friend ends up in possession of the message composed by their counterpart.

We consider an outside observer, Wigner, who has quantum control over a laboratory containing his friend.
The global Hilbert space consists of six subsystems:
the measured qubit $Q$, a room-label register $R$, the friend's memory register for a qubit measurement outcome $F$, the friend's memory register for their message $M$, a piece of paper $P$, and an auxiliary workspace $A$. The workspace $A$ includes all remaining degrees of freedom of the friend's mind and lab, beyond those represented by the other registers. The register $R$ functions as a classical record indicating which room the friend occupies. It serves as a branch label, analogous to a measurement pointer indicating the friend's location. In the state expressions and circuits below, we use the register ordering $Q,R,F,A,M,P$. The scenario is visualised in Fig.~\ref{fig:branch-swap}.

\begin{figure*}
    \centering
    \includegraphics[width=0.8\linewidth]{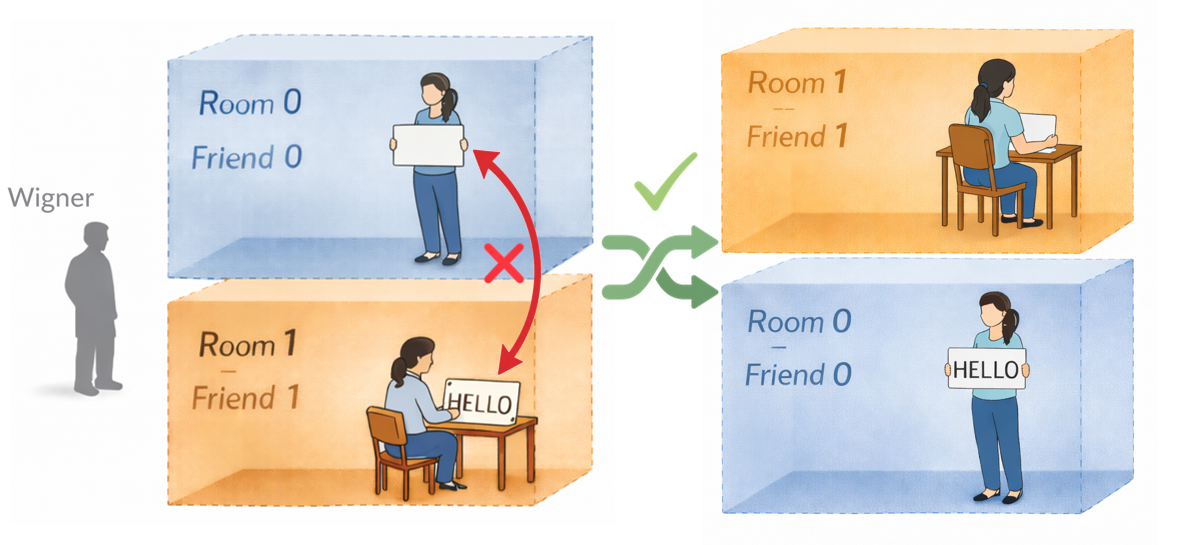}
    \caption{Hello Worlds: Wigner cannot exchange the friends' messages directly, but can instead switch the friends' branches such that one obtains a message from the other.}
    \label{fig:branch-swap}
\end{figure*}

The friend begins with their qubit outcome memory register in the state $\ket{0}_F$, while the qubit to be measured is prepared in the state
\begin{equation} \label{state_1}
\ket{+}_Q = \tfrac{1}{\sqrt{2}}(\ket{0}_Q+\ket{1}_Q).
\end{equation}
After the friend measures the qubit, which can be modelled by a controlled-NOT (\textsc{cnot}) gate controlled on the qubit and targeted on the friend's memory, the joint state of $Q$ and $F$ becomes
\begin{equation} \label{state_2}
\tfrac{1}{\sqrt{2}}(\ket{00}_{QF}+\ket{11}_{QF}).
\end{equation}

The friend was told before the experiment that if they see the outcome 0, they enter a room labeled ``0'', and if they see the outcome 1, they enter a different room labeled ``1''. In our global statevector representation, we account for the friend's location by including the room-label register $R$, and this room record interacts with the friend via a \textsc{cnot} gate controlled on $F$ and targeted on $R$.
Including this record explicitly, the global state becomes
\begin{equation} \label{state_3}
\tfrac{1}{\sqrt{2}}(\ket{000}_{QRF}+\ket{111}_{QRF}).
\end{equation}

We refer to the friend in the $R=0$ branch as friend-0 and in the $R=1$ branch as friend-1. Each branch contains a blank piece of paper, initially in the state $\ket{0}_P$, and the friend's memory register and remaining degrees of freedom begin in the states $\ket{0}_M$ and $\ket{0}_A$, respectively.

Having entered their room, the friend has an experience of it and composes a message. We model the dynamics of this process, including the perception of the room and the thinking process that results in a message, by a unitary $V$ evolving the laboratory registers $F$, $A$ and $M$:

\begin{equation} \label{Vdef}
V \ket{f}_F \ket{0}_A \ket{0}_M = \ket{f}_F \ket{\gamma_f}_A \ket{\mu_f}_M ,
\end{equation}
where $\mu_f$ denotes the classical message composed by the friend who observed outcome $f$, and we write $\ket{\mu_f}_M$ for its encoding in the memory register. The state $\ket{\gamma_f}_A$ represents the effects of the message-writing process on the friend's mind and lab.

We highlight that $V$ is specified as a physical process, the dynamical evolution of the isolated laboratory, and not by its action on any particular state. Wigner is assumed to be capable of running the dynamics of the laboratory forwards and backwards, which is to say that he can apply $V$ and $V^\dagger$. He is not assumed to know the content of the messages that the friends come up with. Knowledge of dynamics (e.g. specified by a unitary operation) does not by itself imply knowledge of the output state from applying those dynamics to some input state. 

Including the blank paper subsystem, the global state after applying $V$ to Eq.~\eqref{state_3} is
\begin{equation} \label{m_state_1}
\begin{aligned}
&\tfrac{1}{\sqrt{2}}
\bigl(
\ket{000}_{QRF}\ket{\gamma_0}_A\ket{\mu_0}_M\\
&\qquad +
\ket{111}_{QRF}\ket{\gamma_1}_A\ket{\mu_1}_M
\bigr)\ket{0}_P ,
\end{aligned}
\end{equation}
and the friends then write their messages (copying them from their minds to the paper) via a \textsc{cnot} from $M$ to $P$:
\begin{equation} \label{m_state_2}
\begin{aligned}
&\tfrac{1}{\sqrt{2}}
\bigl(
\ket{000}_{QRF}\ket{\gamma_0}_A\ket{\mu_0}_M\ket{\mu_0}_P\\
&\qquad +
\ket{111}_{QRF}\ket{\gamma_1}_A\ket{\mu_1}_M\ket{\mu_1}_P
\bigr).
\end{aligned}
\end{equation}

Wigner now uncomputes the diverging experiences of the friends in the two rooms by applying $V^\dagger$, while leaving the paper subsystem untouched:
\begin{equation} \label{m_state_3}
\tfrac{1}{\sqrt{2}}
\bigl(
\ket{00000}_{QRFAM}\ket{\mu_0}_P
+
\ket{11100}_{QRFAM}\ket{\mu_1}_P
\bigr).
\end{equation}
At this point the friends have been returned exactly to the states they occupied on entering their rooms, with both $M$ and $A$ reset to their initial states. Now we make a key observation: at this point, just after the friends have split into different branches but before they have started having diverging experiences in those branches, Wigner knows the precise difference between the two friends' states. Apart from the paper record, the only degrees of freedom distinguishing the branches are those in the friends' minds and measurement devices recording the measurement outcome. Wigner can then apply a global unitary that exchanges all degrees of freedom distinguishing the two branches except for the paper. In the abstraction we have used here for the degrees of freedom of the distinguishing features of the branches, this operation is
\begin{equation} \label{state_7}
U_{\rightleftharpoons} = X_Q \otimes X_R \otimes X_F,
\end{equation}
which swaps the friend, the measured qubit state, and the corresponding room-label record while acting trivially on the message memory, paper (which keeps a record of the written message), and other all degrees of freedom of the friend and lab. We refer to this as a partial branch-swap operation, depicted by a unitary with a subscript of opposing half-arrows $\rightleftharpoons$ to represent the partial exchange of branches. 

Wigner applies the partial branch-swap $U_{\rightleftharpoons}$ of Eq.~\eqref{state_7}:
\begin{equation} \label{m_state_4}
\tfrac{1}{\sqrt{2}}
\bigl(
\ket{11100}_{QRFAM}\ket{\mu_0}_P
+
\ket{00000}_{QRFAM}\ket{\mu_1}_P
\bigr),
\end{equation}
and finally recomputes by applying $V$ once more, yielding the state:
\begin{equation} \label{m_state_5}
\begin{aligned}
&\tfrac{1}{\sqrt{2}}
\bigl(
\ket{111}_{QRF}\ket{\gamma_1}_A\ket{\mu_1}_M\ket{\mu_0}_P\\
&\qquad +
\ket{000}_{QRF}\ket{\gamma_0}_A\ket{\mu_0}_M\ket{\mu_1}_P
\bigr).
\end{aligned}
\end{equation}

A quantum circuit depiction of the process is given in Fig.~\ref{fig:memory-protocol-circuit}. We encapsulate the possibility of this protocol by Theorem \ref{thm:memory}.

\begin{figure}[h]
\centering
\begin{quantikz}[row sep=0.2cm]
\ket{+}_Q & \ctrl{2} & \qw & \qw & \qw & \qw & \gate{X}
\gategroup[wires=3,steps=1,
    style={dashed,rounded corners,inner xsep=2pt,inner ysep=3pt}]{} & \qw & \qw \\
\ket{0}_R & \qw & \targ{} & \qw & \qw & \qw & \gate{X} & \qw & \qw \\
\ket{0}_F & \targ{} & \ctrl{-1} & \gate[wires=3]{V} & \qw & \gate[wires=3]{V^\dagger} & \gate{X} & \gate[wires=3]{V} & \qw \\
\ket{0}_A & \qw & \qw & & \qw & & \qw & & \qw \\
\ket{0}_M & \qw & \qw & & \ctrl{1} & & \qw & & \qw \\
\ket{0}_P & \qw & \qw & \qw & \targ{} & \qw & \qw & \qw & \qw
\end{quantikz}
\caption{Quantum circuit implementing the inter-branch message-transfer protocol.
The room-label register $R$ records the friend's location via a \textsc{cnot} from $F$ to $R$.
The gate $V$ is the dynamics of the friend in the lab thinking of the messages, including the additional cognitive and lab state stored in $A$. The dynamics are run forwards, reversed, and run forwards again, with the branch exchange (dashed box) applied in between.}
\label{fig:memory-protocol-circuit}
\end{figure}

\begin{theorem}[Inter-branch message transfer]\label{thm:memory}
There exists a global quantum operation that leaves each of two observers in distinct branches with a written record of the message composed by the observer in the other branch, without requiring prior knowledge of the
message.
\end{theorem}

\begin{proof}
By explicit construction, from Eq.~\eqref{state_3} through
Eqs.~\eqref{m_state_1}--\eqref{m_state_5}. The protocol implements the transformation
\begin{equation}
\begin{aligned}
&\frac{1}{\sqrt{2}}
\sum_{f\in\{0,1\}}
\ket{fff}_{QRF}\ket{\gamma_f}_A\ket{\mu_f}_M\ket{\mu_f}_P
\\
&\qquad\longrightarrow
\frac{1}{\sqrt{2}}
\sum_{g\in\{0,1\}}
\ket{ggg}_{QRF}\ket{\gamma_g}_A\ket{\mu_g}_M
\ket{\mu_{1-g}}_P .
\end{aligned}
\end{equation}
where we have relabelled the branch index after the exchange according to
$g=1-f$. Thus, in each final branch $g$, the friend has memory
$\ket{\mu_g}_M$ of the message they composed and the restored workspace state of their mind and lab $\ket{\gamma_g}_A$, while the paper contains
$\ket{\mu_{1-g}}_P$, the message composed in the other branch. Hence each
observer is left with a written record of the message composed by the
observer in the other branch.
\end{proof}

The protocol generalises straightforwardly to $n$-bit classical messages by promoting the memory register $M$ and paper $P$ to $n$-qubit systems, with $V$ acting on $F$, $A$ and the $n$ memory qubits. The \textsc{cnot}s between the memory and paper occur transversally between the $n$ qubits in each state (meaning that there is a \textsc{cnot} between the first qubit in each set, one between the second qubit, etc). The resulting generalised circuit is that in Fig.~\ref{fig:memory-protocol-circuit-n}. Each instance of the recomputed friend then remembers composing their own message, but finds their counterpart's written on the paper in front of them.

\begin{figure}[h]
\centering
\begin{quantikz}[row sep=0.2cm]
\ket{+}_Q & \ctrl{2} & \qw & \qw & \qw & \qw & \gate{X}
\gategroup[wires=3,steps=1,
    style={dashed,rounded corners,inner xsep=2pt,inner ysep=3pt}]{} & \qw & \qw \\
\ket{0}_R & \qw & \targ{} & \qw & \qw & \qw & \gate{X} & \qw & \qw \\
\ket{0}_F
  & \targ{}
  & \ctrl{-1}
  & \gate[wires=3]{V}
  & \qw
  & \gate[wires=3]{V^\dagger}
  & \gate{X}
  & \gate[wires=3]{V}
  & \qw \\
\ket{0}_A & \qw & \qw & & \qw & & \qw & & \qw \\
\ket{0}_M^{\otimes n}
  & \qwbundle{n}
  & \qw
  &
  & \ctrl{1}
  &
  & \qw
  &
  & \qw \\
\ket{0}_P^{\otimes n}
  & \qwbundle{n}
  & \qw
  & \qw
  & \targ{}
  & \qw
  & \qw
  & \qw
  & \qw
\end{quantikz}
\caption{Generalisation of the protocol to $n$-qubit messages.}
\label{fig:memory-protocol-circuit-n}
\end{figure}

 Note that the initial part of the thought experiment has the structure of Bennett's reversible computation scheme \cite{bennett1973logical}: compute, copy the result to an output register, then uncompute the computation, to restore the machine to its original condition while leaving the output in place. For multiverse communication, we include two additional steps: apply a branch exchanging operation after the uncomputation, and then run the computation forwards once more.

 Appendix~\ref{sec:original} outlines an earlier form of the protocol with a different implementation, and explains how the need to track all the effects of message-writing on the friend's mind and lab motivates the uncomputation version used here.

\subsection{Message-independence of Wigner's operations}\label{sec:constraints}

The most direct route to inter-branch communication would be for Wigner to apply an operation that exchanges the contents of the two pieces of paper, so the messages appear in the other friends' branches. However, this is impossible without access to unitaries which are parametrised by the message values. We prove this is the case in Proposition~\ref{prop:swap}.

\begin{proposition}[No message-independent message swap] \label{prop:swap}
There is no single fixed unitary operation $G$ such that
$G\ket{\mu_0} \propto \ket{\mu_1}$ and $G\ket{\mu_1} \propto \ket{\mu_0}$
for every pair of messages $\mu_0,\mu_1$ that the friends might compose.
\end{proposition}

\begin{proof}
Suppose such a $G$ existed. Write $m$ for the message composed by friend-0 and $m'$ for that composed by
friend-1, so that the requirement is $G\ket{m} \propto \ket{m'}$ and
$G\ket{m'} \propto \ket{m}$. Applying it to the case in which both compose the same message gives
\begin{equation}
G\ket{m} \propto \ket{m} \qquad \text{for every message } m ,
\end{equation}
so that $G$ is diagonal in the message basis. Applying it to a case in which
they compose different messages $m \neq m'$ gives
\begin{equation}
G\ket{m} \propto \ket{m'}.
\end{equation}
Distinct classical
messages are represented by orthogonal states, so $G\ket{m}$ cannot be proportional to both $\ket{m}$ and $\ket{m'}$.
Therefore no such $G$ exists.
\end{proof}

This is the sense in which the protocol can exchange the observers across the branches but not directly exchange the messages across the branches. This contrasts with the operation $V$: although $V$ is the dynamics that generates the messages, it is fixed independently of their realised values. The protocol requires control of $V$ and $V^\dagger$, but does not require the message values to be determined first and then used to parametrise a new operation.

\section{Discussion}

\subsection{Testing many- vs single-world theories using a knowledge-paradox}\label{sec:knowledge}

There are thought experiments involving modifications to quantum mechanics that enable the emergence of knowledge paradoxes, with a paradigmatic example being in the context of closed-timelike-curves (i.e. time-travel) \cite{Deutsch1991,Bennett2009,Aaronson2009,Lloyd2011}. By imposing a consistency condition, it is possible to resolve Grandfather paradoxes, which are logically inconsistent timelines such as where you travel back in time and kill your grandfather before you were born \cite{Deutsch1991}. However, the condition does not by itself straightforwardly resolve \textit{knowledge paradoxes}. These are time-travel paradoxes in which someone copies a mathematical proof from the future, and that copied proof then becomes the very proof that exists in the future \cite{Deutsch1991}. Such scenarios are logically consistent, and compatible with the consistency condition, but there is no explanation for where the knowledge of the proof came from; hence a knowledge paradox.

We can construct a similar form of paradox within the inter-branch communication protocol, such that the generation of knowledge provides a test of the multiverse account of quantum theory against single-world accounts. Consider participating in the experiment, from the perspective of friend-0. After Wigner performs the uncomputation, branch swap and recomputation, the friend remembers composing $\mu_0$, and yet finds the distinct message $\mu_1$ written on the paper in front of her. In principle, this message could contain new knowledge, such as a novel mathematical proof. Friend-0 requires an explanation for how that knowledge was created, and if the only known mechanism for generating knowledge is another agent, then it provides evidence that friend-1 physically existed in a form that is just as real as friend-0.

This test has a similar spirit to Deutsch's 1985 thought experiment for testing unitary quantum theory (including the Everettian interpretation) against theories where measurement causes an irreversible collapse (such as objective collapse theories) \cite{Deutsch1985}. However the inter-branch communication test may distinguish between a different class of rival theories. On the assumption that the creation of new knowledge requires an experiencing agent, explaining the generation of such knowledge requires the real existence of friend-1 in the other branch, not just unitary dynamics where measurements are reversible. Thus, the former thought experiment distinguishes collapse and no-collapse theories, while the latter distinguishes single-world and many-world theories (more precisely: worlds where all branches of the quantum state are physically real and can thus sustain knowledge-creating observers, and worlds where multiple observers in superposition do not all physically exist). An example may be to distinguish many-worlds and Bohmian interpretations of quantum mechanics \cite{Goldstein2017}, which are indistinguishable by many conceivable empirical tests due to both having unitary dynamics and being no-collapse theories.

\subsection{Does inter-branch communication violate linearity?}

Deterministic nonlinear modifications of quantum mechanics provide a useful point of comparison, such as those introduced by Weinberg \cite{Weinberg1989,Weinberg1989b}. Nonlinear dynamics can make the evolution of a subsystem depend on the global state decomposition rather than solely on its reduced density matrix, and hence local operations can influence observables associated with other decohered branches. Polchinski showed that if Weinberg-type nonlinear dynamics are constrained to forbid superluminal signalling, they generically permit communication between decohered Everettian branches, a phenomenon he termed an ``Everett phone'' \cite{Polchinski1991}. In that protocol, the preparation and reception of a message occur in different branches of the observer's wavefunction.

Polchinski emphasised that such inter-branch communication may be conceptually even more unsettling than faster-than-light signalling, remarking that ``\textit{communication between branches of the wave function seems even more bizarre than faster-than-light communication \ldots\ but it is not clear that it represents an actual inconsistency.}'' Polchinski's result suggested that inter-branch communication was a possible pathology of nonlinear dynamics; our protocol shows that this bizarre form of information transfer can be accessed even within strictly linear quantum mechanics, given sufficient global control. 

\subsection{Predictability and identity}

We have discussed various constraints on the multiverse communication protocol based on the requirement that Wigner's operations are not parametrised by the realised message values. One might imagine a potential loophole: could Wigner deterministically predict the contents of the messages which friend-0 and friend-1 come up with? 
On the assumption that knowledge-creation is fundamentally unpredictable, the friends can write messages that Wigner cannot predict, even if he knew the full details of their initial state. We could similarly consider friend-0 attempting to predict the contents of friend-1's message, and vice-versa, thus directly bypassing the need for communicating between branches in order to retrieve the message. Again, if we assume that friend-1 has had a different experience in her branch to friend-0, then friend-1's knowledge creation becomes unpredictable to friend-0. Even if it was somehow the case that Wigner and friend-0 could deterministically predict the contents of friend-1's message, this would not falsify the inter-branch communication protocol. Instead it would render its existence trivial, considering all three agents involved can deduce the contents of the message with no inter-branch communication required.

Here we highlight that for Wigner to apply $V$ and $V^\dagger$, he must know the laboratory's dynamics; he is not required to know, or to be able to predict, what those dynamics compute. A reversible program can be run and reversed by an agent ignorant of its output, meaning that Wigner is able to run the protocol without ever finding out or predicting the content of the friends' messages.

The protocol also raises philosophical questions regarding the continuity of identity. Here we identified friend-0 and friend-1 with their abstract information content represented by the $\ket{0}$ and $\ket{1}$ states. However an alternative interpretation of the protocol could be that the swapping operation is actually transforming friend-0 into friend-1, and friend-1 into friend-0, rather than moving the friends (and by proxy, the messages) between branches. These interpretations are indistinguishable at this level of abstraction, creating an interesting problem to analyse in terms of the philosophy of identity across multiverse branches, and how the precise implementation of the swapping operation may play a role (if at all). The question is sharpened by the fact that the friends are returned to states that only differ in their record of the measurement outcome before the exchange is applied. At that moment, they are as similar as two observers in distinct branches can be while remaining distinct. What is not ambiguous, on any reading, is the causal history: $\mu_0$ was produced by dynamics acting on a friend who had the $R=0$ experience, and it ends up received by a friend who had the $R=1$ experience.

\subsection{Implementation, toy-models and extensions}

One may consider what the difficulty is of actually implementing such an experiment. A simple qubit simulation of the protocol can likely already be done on today's quantum computers, with the simplest version requiring only six qubits and a shallow circuit consisting of a small number of standard gates. To become more realistic, the size of the message can be increased, requiring a more complex encoding unitary and reversal of the friend's message-generating dynamics.

In the near term, the inter-branch communication protocol can be explored further using the approach developed in related pedagogical work by the author, which demonstrates how quantum thought experiments involving observers can be encoded as explicit quantum circuits, and varied to probe conceptual features of quantum paradoxes on existing quantum hardware \cite{Violaris2023PhysicsLab, Violaris2024EntanglingDisciplines}. There is also growing interest in the scientific community to extend beyond such toy models and implement increasingly realistic Wigner's-friend--type experiments on quantum information processing platforms, including explicit proposals towards encoding an observer within a quantum computer \cite{Zeng2025,Bong2020,WisemanCavalcantiRieffel2023}. These efforts suggest a pathway by which inter-branch communication protocols of the type proposed here may ultimately be explored experimentally.

An interesting direction for future work is to explore whether the exchange can be extended to a continuous dialogue, with later messages depending on earlier ones. We can also consider the possibilities of communication between observers across a greater number of multiverse branches, and whether the Born weights of messages can be redistributed by Wigner. More generally, there are many avenues for characterising and quantifying multiverse communication, and extending the protocol to different regimes.

\section{Conclusion}\label{sec:conclusion}
Using a Wigner's friend scenario, we have presented a protocol that exchanges classical messages between observers in distinct Everett branches by swapping observer states. By uncomputing the observers' dynamics before the exchange and recomputing them afterwards, each observer receives a message composed by their counterpart in the other branch. The protocol is fully consistent with standard unitary quantum mechanics, and therefore challenges intuitive notions of information locality in the multiverse and invites refined theoretical, philosophical, and experimental investigation.

\section*{Acknowledgements}

I thank David Deutsch, Sam Kuypers, Ryan Mann and Tony Short for helpful discussions and comments, including suggestions regarding the interpretation of the protocol and the distinction between knowing a physical process and knowing what it computes. I also thank Chiara Marletto, Vlatko Vedral, and the other participants of the Cocconato Constructor Theory Workshop 2024 for valuable discussions. Finally I thank Rob Sullivan for discussions that included posing the question that motivated this work.

\bibliographystyle{apsrev4-2-titles}
\bibliography{references}

@article{Everett1957,
  author  = {Everett, Hugh},
  title   = {Relative State Formulation of Quantum Mechanics},
  journal = {Reviews of Modern Physics},
  volume  = {29},
  number  = {3},
  pages   = {454--462},
  year    = {1957},
  doi     = {10.1103/RevModPhys.29.454}
}

@article{Deutsch2002,
  author  = {Deutsch, David},
  title   = {The Structure of the Multiverse},
  journal = {Proceedings of the Royal Society A},
  volume  = {458},
  number  = {2028},
  pages   = {2911--2923},
  year    = {2002},
  doi     = {10.1098/rspa.2002.1015}
}

@book{Wallace2012,
  author    = {Wallace, David},
  title     = {The Emergent Multiverse: Quantum Theory according to the Everett Interpretation},
  publisher = {Oxford University Press},
  year      = {2012},
  address   = {Oxford}
}

@article{Zurek2003,
  author  = {Zurek, Wojciech H.},
  title   = {Decoherence, Einselection, and the Quantum Origins of the Classical},
  journal = {Reviews of Modern Physics},
  volume  = {75},
  number  = {3},
  pages   = {715--775},
  year    = {2003},
  doi     = {10.1103/RevModPhys.75.715}
}

@book{Deutsch1997,
  author    = {Deutsch, David},
  title     = {The Fabric of Reality},
  publisher = {Allen Lane},
  year      = {1997},
  address   = {London}
}

@article{Polchinski1991,
  author  = {Polchinski, Joseph},
  title   = {{Weinberg's Nonlinear Quantum Mechanics and the Einstein--Podolsky--Rosen Paradox}},
  journal = {Physical Review Letters},
  volume  = {66},
  number  = {4},
  pages   = {397--400},
  year    = {1991},
  doi     = {10.1103/PhysRevLett.66.397}
}

@article{Schlosshauer2005,
  author  = {Schlosshauer, Maximilian},
  title   = {Decoherence, the Measurement Problem, and Interpretations of Quantum Mechanics},
  journal = {Reviews of Modern Physics},
  volume  = {76},
  number  = {4},
  pages   = {1267--1305},
  year    = {2005},
  doi     = {10.1103/RevModPhys.76.1267}
}

@article{Riedel2017,
  title={The rise and fall of redundancy in decoherence and quantum Darwinism},
  author={Riedel, C Jess and Zurek, Wojciech H and Zwolak, Michael},
  journal={New Journal of Physics},
  volume={14},
  number={8},
  pages={083010},
  year={2012},
  publisher={IOP Publishing}
}

@article{Wigner1961,
  title={Remarks on the mind-body question},
  author={Wigner, Eugene P},
  journal={Philosophical reflections and syntheses},
  pages={247--260},
  year={1995},
  publisher={Springer}
}

@article{frauchiger2018,
  title={Quantum theory cannot consistently describe the use of itself},
  author={Frauchiger, Daniela and Renner, Renato},
  journal={Nature communications},
  volume={9},
  number={1},
  pages={3711},
  year={2018},
  publisher={Nature Publishing Group UK London}
}

@article{Brukner2018,
  title={A no-go theorem for observer-independent facts},
  author={Brukner, {\v{C}}aslav},
  journal={Entropy},
  volume={20},
  number={5},
  pages={350},
  year={2018},
  publisher={MDPI}
}

@article{Bong2020,
  title={{A strong no-go theorem on the Wigner’s friend paradox}},
  author={Bong, Kok-Wei and Utreras-Alarc{\'o}n, An{\'\i}bal and Ghafari, Farzad and Liang, Yeong-Cherng and Tischler, Nora and Cavalcanti, Eric G and Pryde, Geoff J and Wiseman, Howard M},
  journal={Nature Physics},
  volume={16},
  number={12},
  pages={1199--1205},
  year={2020},
  publisher={Nature Publishing Group UK London}
}

@article{Aaronson2009,
  title={Closed timelike curves make quantum and classical computing equivalent},
  author={Aaronson, Scott and Watrous, John},
  journal={Proceedings of the Royal Society A: Mathematical, Physical and Engineering Sciences},
  volume={465},
  number={2102},
  pages={631--647},
  year={2009},
  publisher={The Royal Society London}
}

@article{Lloyd2011,
  title={Quantum mechanics of time travel through post-selected teleportation},
  author={Lloyd, Seth and Maccone, Lorenzo and Garcia-Patron, Raul and Giovannetti, Vittorio and Shikano, Yutaka},
  journal={Physical Review D—Particles, Fields, Gravitation, and Cosmology},
  volume={84},
  number={2},
  pages={025007},
  year={2011},
  publisher={APS}
}

@article{Deutsch1991,
  title={Quantum mechanics near closed timelike lines},
  author={Deutsch, David},
  journal={Physical Review D},
  volume={44},
  number={10},
  pages={3197},
  year={1991},
  publisher={APS}
}

@article{Bennett2009,
  title={Can closed timelike curves or nonlinear quantum mechanics improve quantum state discrimination or help solve hard problems?},
  author={Bennett, Charles H and Leung, Debbie and Smith, Graeme and Smolin, John A},
  journal={{Physical Review Letters}},
  volume={103},
  number={17},
  pages={170502},
  year={2009},
  publisher={APS}
}

@article{Deutsch1985,
  title={Quantum theory as a universal physical theory},
  author={Deutsch, David},
  journal={International Journal of Theoretical Physics},
  volume={24},
  number={1},
  pages={1--41},
  year={1985},
  publisher={Springer}
}

@misc{Goldstein2017,
  author       = {Goldstein, Sheldon},
  title        = {Bohmian Mechanics},
  howpublished = {\emph{Stanford Encyclopedia of Philosophy}},
  editor       = {Zalta, Edward N.},
  year         = {2017},
  note         = {\url{https://plato.stanford.edu/entries/qm-bohm/}}
}

@article{Zeng2025,
  title={Towards violations of Local Friendliness with quantum computers},
  author={Zeng, William J and Labib, Farrokh and Russo, Vincent},
  journal={Quantum},
  volume={9},
  pages={1851},
  year={2025},
  publisher={Verein zur F{\"o}rderung des Open Access Publizierens in den Quantenwissenschaften}
}

@article{WisemanCavalcantiRieffel2023,
  title={A" thoughtful" Local Friendliness no-go theorem: a prospective experiment with new assumptions to suit},
  author={Wiseman, Howard M and Cavalcanti, Eric G and Rieffel, Eleanor G},
  journal={Quantum},
  volume={7},
  pages={1112},
  year={2023},
}

@article{Violaris2024EntanglingDisciplines,
  author    = {Violaris, Maria},
  title     = {Entangling Disciplines: Causality, Entropy and Time-Travel Paradoxes on a Quantum Computer},
  journal = {2024 IEEE Quantum Science and Engineering Education Conference (QSEEC)},
  year      = {2024},
  pages     = {71--81},
}

@article{Violaris2023PhysicsLab,
  author    = {Violaris, Maria},
  title     = {A Physics Lab Inside Your Head: Quantum Thought Experiments as an Educational Tool},
  journal = {2023 IEEE International Conference on Quantum Computing and Engineering (QCE)},
  year      = {2023},
  pages     = {203--213},
}

@article{weinberg1989,
  title={Precision tests of quantum mechanics},
  author={Weinberg, Steven},
  journal={{Physical Review Letters}},
  volume={62},
  number={5},
  pages={485},
  year={1989},
  publisher={APS}
}

@article{weinberg1989b,
  title={Testing quantum mechanics},
  author={Weinberg, Steven},
  journal={Annals of Physics},
  volume={194},
  number={2},
  pages={336--386},
  year={1989},
  publisher={Elsevier}
}

@article{bennett1973logical,
  title={Logical reversibility of computation},
  author={Bennett, Charles H},
  journal={IBM journal of Research and Development},
  volume={17},
  number={6},
  pages={525--532},
  year={1973},
  publisher={IBM}
}

@article{del2025wigner,
  title={Wigner’s friend scenarios: On what to condition and how to verify the predictions},
  author={Del Santo, Flavio and Manzano, Gonzalo and Brukner, {\v{C}}aslav},
  journal={Physical Review Research},
  volume={7},
  number={3},
  pages={033279},
  year={2025},
  publisher={APS}
}

@incollection{baumann2020wigner,
  title={Wigner’s friend as a rational agent},
  author={Baumann, Veronika and Brukner, {\v{C}}aslav},
  booktitle={Quantum, probability, logic: the work and influence of Itamar Pitowsky},
  pages={91--99},
  year={2020},
  publisher={Springer}
}

\appendix

\section{The original form of the protocol}\label{sec:original}

This appendix records an earlier version of the protocol, in which the observer who writes the message cannot retain memory of it after it transfers to the observer in the other branch. This version is useful to demonstrate the subtleties of the operations allowed by Wigner and the consequences of these. The toy circuit works under the idealisation that composing a message changes only the message memory. Its limitation is that a general cognitive process also changes other degrees of freedom in the friend's mind and potentially in the lab, represented by $A$ in the main text and omitted in this idealised circuit; resetting the message memory alone does not reset those degrees of freedom. The uncomputation process used in the main text addresses this limitation. For simplicity, here we have a single message being written by friend-1, for communication to friend-0.

The friend is instructed to write a message only if they observe outcome 1. Starting from the state of Eq.~\eqref{state_3}, with $M$ initially in $\ket{0}_M$, we represent the process of thinking of the message as a controlled-$\mu$ gate (controlled on $F$ and targeted on $M$), giving
\begin{equation} \label{state_4}
\tfrac{1}{\sqrt{2}}
\bigl(
\ket{0000}_{QRFM}
+
\ket{111\mu}_{QRFM}
\bigr),
\end{equation}
Friend-1 then writes the message she has thought of on the paper, modelled via a \textsc{cnot} gate controlled on $M$ and targeted on $P$:
\begin{equation} \label{state_5}
\tfrac{1}{\sqrt{2}}
\bigl(
\ket{0000}_{QRFM}\ket{0}_P
+
\ket{111\mu}_{QRFM}\ket{\mu}_P
\bigr).
\end{equation}

Wigner now clears the memory register by applying a \textsc{cnot} from $P$ to $M$. Since the two registers hold the same value, this returns $M$ to $\ket{0}_M$ in both branches, independently of the message state. The resulting global state is
\begin{equation} \label{state_6}
\tfrac{1}{\sqrt{2}}
\bigl(
\ket{0000}_{QRFM}\ket{0}_P
+
\ket{1110}_{QRFM}\ket{\mu}_P
\bigr).
\end{equation}

Applying the partial branch-swap of Eq.~\eqref{state_7}, the global state becomes
\begin{equation} \label{state_8}
\tfrac{1}{\sqrt{2}}
\bigl(
\ket{1110}_{QRFM}\ket{0}_P
+
\ket{0000}_{QRFM}\ket{\mu}_P
\bigr).
\end{equation}
At this point, the paper in the $R=0$ branch contains the message $\mu$, while the paper in the $R=1$ branch is blank. The circuit is shown in Fig.~\ref{fig:original-circuit}. The result is that friend-0 has been switched to the branch of friend-1, receiving the message from friend-1; meanwhile friend-1 has switched to the branch of friend-0, and has no message and no memory of what the message that she wrote contained. 

\begin{figure}[h]
\centering
\begin{quantikz}[row sep=0.2cm]
\ket{+}_Q & \ctrl{2} & \qw & \qw & \qw & \qw & \gate{X} 
\gategroup[wires=3,steps=1,
    style={dashed,rounded corners,inner xsep=2pt,inner ysep=3pt}]{} & \qw \\
\ket{0}_R & \qw & \targ{} & \qw & \qw & \qw & \gate{X} & \qw \\
\ket{0}_F & \targ{} & \ctrl{-1} & \ctrl{1} & \qw & \qw & \gate{X} & \qw \\
\ket{0}_M & \qw & \qw & \gate{\mu} & \ctrl{1} & \targ{} & \qw & \qw \\
\ket{0}_P & \qw & \qw & \qw & \targ{} & \ctrl{-1} & \qw & \qw
\end{quantikz}
\caption{Quantum circuit for the version of the inter-branch message-transfer protocol where the friend writing a message does not retain memory of it.}
\label{fig:original-circuit}
\end{figure}

\textbf{Why this version is superseded—} In this version of the thought experiment, the only effect of thinking of the message on the state is a record of the message being stored in the memory register. This can then be returned to its initial state through a message-independent \textsc{cnot} operation from the paper. However in general, the state of the brain of the friend after the thinking process will be different in other degrees of freedom, as well as the memory storage of the message content. In the notation of the main text, the workspace can retain a branch-dependent state $\ket{\gamma_f}_A$ even after the message memory has been reset. Wigner cannot apply an operation to directly switch between the two friend branches after this thinking process, even if the memory is erased, because he does not have knowledge of how the friends in the two branches now differ (which is the information required to construct a branch-swap operation). 

Since Wigner must be able to specify the degrees of freedom differentiating the observers in the two branches in order to swap them, he must undo the whole of the friend's message-generating dynamics.

That is what $V^\dagger$ accomplishes in Sec.~\ref{sec:memory-protocol}. It resets both the message memory $M$ and the remaining degrees of freedom of the friend and lab $A$, leaving the friend in a state that depends on the branch only through the recorded measurement outcome: a degree of freedom that Wigner himself prepared, whose location he knows. The branch swap is then well-defined and running the dynamics forwards again after the exchange restores the friends complete with their memories. An interesting development of the new formulation of the thought experiment in the main text is that the observers can retain their memory of the message. 

\end{document}